\documentclass[twocolumn,twoside]{IEEEtran} 

\usepackage{amsmath}
\usepackage{amssymb}
\usepackage{array}
\usepackage{fixmath}
\usepackage{graphicx}
\usepackage{cite}
\usepackage{color}
\usepackage{changes}
\usepackage{algorithm}
\usepackage{algcompatible}
\usepackage{blindtext}
\usepackage{acronym}
\usepackage{soul}

\usepackage{amsthm}

\DeclareMathAlphabet{\mathbit}{OML}{cmr}{bx}{it}

\newcommand{\B}[1]{{\bf{#1}}}
\DeclareMathOperator{\Transpose}{T}
\DeclareMathOperator{\Hermitian}{H}
\newcommand{\Tr}{{\Transpose}}
\newcommand{\He}{{\Hermitian}}
\DeclareMathOperator{\Exp}{\mathbb{E}}

\DeclareMathOperator{\trace}{tr}

\DeclareMathOperator{\Real}{Re}
\DeclareMathOperator{\Imag}{Im}
\DeclareMathOperator*{\argmax}{argmax}

\newcommand{\eul}{{\text{e}}}
\newcommand{\imj}{{\text{j}}}
\newcommand{\hs}{\bar{\B{h}}}
\newcommand{\hns}{\tilde{\B{h}}}

\usepackage[numbers,sort&compress]{natbib}

\newtheorem{lemma}{Lemma}

\acrodef{ADC}[ADC]{Analog-to-Digital Converter}
\acrodef{AO}[AO]{Alternating Optimization}
\acrodef{AoA}[AoA]{Angle of Arrival}
\acrodef{AoD}[AoD]{Angle of Departure}
\acrodef{APN}[APN]{Analog Precoding Network}
\acrodef{ASK}[ASK]{Amplitude-Shift Keying}
\acrodef{AWGN}[AWGN]{Additive White Gaussian Noise}
\acrodef{BER}[BER]{Bit Error Ratio}
\acrodef{BF}[BF]{Beamforming}
\acrodef{BFN}[BFN]{Beamforming Network}
\acrodef{BS}[BS]{Base Station}
\acrodef{CB}[CB]{conjugate beamforming}
\acrodef{COMP}[COMP]{Covariance OMP}
\newacro{CSI}{Channel State Information}
\acrodef{DAC}[DAC]{Digital to Analog Converter}
\acrodef{DBS}[DBS]{Distance-Based Scheduling}
\acrodef{DCS}[DCS]{Digital Communication System}
\acrodef{DCOMP}[DCOMP]{Dynamic COMP}
\newacro{DFT}{Discrete Fourier Transform}
\acrodef{DL}[DL]{downlink}
\acrodef{DOA}[DOA]{Direction Of Arrival}
\acrodef{DPC}[DPC]{dirty-paper coding}
\acrodef{ELAA}[ELAA]{extremely-large aperture array}
\acrodef{ESD}[ESD]{Energy Spectral Density}
\newacro{FDD}{Frequency-Division Duplex}
\acrodef{FSK}[FSK]{Frequency-Shift Keying}
\acrodef{FT}[FT]{Fourier Transform}
\acrodef{HT}[HT]{Hilbert Transform}
\acrodef{HI}[HI]{Harmonic Interference}
\acrodef{ICI}[ICI]{Inter-Carrier Interference}
\acrodef{IL}[IL]{Insertion Losses}
\acrodef{ISI}[ISI]{Inter-Symbol Interference}
\acrodef{IUI}[IUI]{inter-user interference}
\acrodef{JSDM}[JSDM]{Joint Spatial Division and Multiplexing}
\acrodef{LBFN}[LBFN]{Linear Beamforming Network}
\acrodef{LLF}[LLF]{Log-Likelihood function}
\acrodef{LMD}[LMD]{Linearly Modulated Digital}
\acrodef{LoS}[LoS]{line-of-sight}
\newacro{MAP}{Maximum A Posteriori}
\acrodef{MIMO}[MIMO]{Multiple-Input Multiple-Output}
\newacro{ML}{Maximum Likelihood}
\newacro{MMSE}{Minimum Mean Squared Error}
\acrodef{MMV}[MMV]{Multiple Measurement Vector}
\acrodef{mmWave}[mmWave]{millimeter wave}
\acrodef{MRC}[MRC]{Maximum Ratio Combining}
\acrodef{MRT}[MRT]{Maximum Ratio Trasmission}
\newacro{MSE}{Mean Squared Error}
\acrodef{MUSIC}[MUSIC]{MUltiple SIgnal Classification}
\acrodef{NLoS}[NLoS]{Non Line-of-sight}
\acrodef{NMSE}{Normalized Mean Squared Error}
\acrodef{OFDM}[OFDM]{Orthogonal Frequency-Division Multiplexing}
\acrodef{OFDMA}[OFDMA]{Orthogonal Frequency-Division Multiple Access}
\acrodef{OMP}[OMP]{Orthogonal Matching Pursuit}
\acrodef{OSMP}[OSMP]{Orthogonal Subspace Matching Pursuit}
\acrodef{PA}[PA]{Power Amplifier}
\acrodef{PS}[PS]{Phase Shifter}
\acrodef{PSK}[PSK]{Phase-Shift Keying}
\acrodef{QAM}[QAM]{Quadrature Amplitude Modulation}
\acrodef{RF}[RF]{Radio Frequency}
\acrodef{RFC}[RFC]{Rayleigh Fading Channel}
\acrodef{SDMA}[SDMA]{space-division multiple access}
\acrodef{SE}[SE]{sum-spectral efficiency}
\acrodef{SER}[SER]{Symbol Error Rate}
\acrodef{SIC}[SIC]{Successive Interference Cancellation}
\acrodef{SR}[SR]{Sideband Radiation}
\acrodef{SINR}[SINR]{Signal-to-Interference-plus-Noise Ratio}
\acrodef{SLL}[SLL]{Side-Lobe Level}
\acrodef{SOCP}[SOCP]{Second-Order Cone Program}
\acrodef{SOMP}[SOMP]{Simultaneous-Orthogonal Matching Pursuit}
\acrodef{SPDT}[SPDT]{Single-Pole-Double-Throw}
\acrodef{SPST}[SPST]{Single-Pole-Single-Throw}
\acrodef{SR}[SR]{Sideband Radiation}
\acrodef{SS}[SS]{Spatial Smoothing}
\acrodef{SNR}[SNR]{Signal-to-Noise Ratio}
\acrodef{SUS}{semi-orthogonal user selection}
\acrodef{SW}{spherical wavefront}
\newacro{TDD}{Time-Division Duplex}
\acrodef{TM}[TM]{Time Modulation}
\acrodef{TMA}[TMA]{Time-Modulated Array}
\acrodef{ULA}[ULA]{Uniform Linear Array}
\acrodef{UPA}[UPA]{Uniform Planar Array}
\acrodef{VGA}[VGA]{Variable Gain Amplifier}
\acrodef{VPS}[VPS]{Variable Phase Shifter}
\acrodef{VR}[VR]{visibility regions}
\acrodef{XL}[XL]{extra-large}
\acrodef{ZF}[ZF]{Zero-Forcing}

\makeatletter
\def\blfootnote{\xdef\@thefnmark{}\@footnotetext}
\makeatother

\begin{document}

\title{{Joint User Scheduling and Precoding for\\ XL-MIMO Systems with Imperfect CSI}}

\author{José P. González-Coma,
	 F. Javier López-Martínez,~\IEEEmembership{Senior Member,~IEEE,}
Luis Castedo,~\IEEEmembership{Senior Member,~IEEE.}}

\maketitle
\begin{abstract}
%
We propose an algorithm for joint precoding and user selection in multiple-input multiple-output systems with extremely-large aperture arrays, assuming realistic channel conditions and imperfect channel estimates. The use of long-term channel state information (CSI) for user scheduling, and a proper selection of the set of users for which CSI is updated allow for obtaining an improved achievable sum spectral efficiency. We also confirm that the effect of imperfect CSI in the precoding vector design and the cost of training must be taken into consideration for realistic performance prediction.
\end{abstract}
\IEEEpeerreviewmaketitle
\vspace{-2mm}
\vspace{-3mm}
\blfootnote{\noindent Manuscript received Feb. XX, 2023; revised XX, 2023. This work was funded in part by Junta de Andaluc\'ia through grant EMERGIA20-00297, and in part by MCIN/AEI/10.13039/501100011033 through grant PID2020-118139RB-I00. The authors thank the Defense
University Center at the Spanish Naval Academy (CUD-ENM) for the support provided for this research.  This work has been submitted to the IEEE for possible publication. Copyright may be transferred without notice, after which this version may no longer be accessible. J.P. Gonz\'alez-Coma is with Defense University Center at the Spanish Naval Academy, Marín 36920, Spain. F. J. L\'opez-Mart\'inez is with Dept. Signal Theory, Networking and Communications, University of Granada, 18071, Granada (Spain). L. Castedo is with Dept. Computer Engineering \& CITIC Research Center, University of A Coru\~na,  A Coru\~na 15001, Spain.
}
\vspace{-2mm}
\section{Introduction}

The use of \acp{ELAA} to further boost the potential of the massive \ac{MIMO} concept is a key ingredient to enable new use cases and emerging applications in the evolution towards 6G \cite{Emil2019}. In these systems, also referred to as \ac{XL}-\ac{MIMO} systems, the number of antennas at the \ac{BS} is in the order of hundreds, which enables to serve a much larger set of users. One of the crucial challenges when considering the case of XL-MIMO is that, when the user distances to the \ac{BS} are comparable to the \ac{BS} size, the consideration of \ac{SW} propagation is instrumental to exploit the features of this regime \cite{LuZe22,Cui23}. 

The promising attributes of \ac{XL}-\ac{MIMO} systems in multi-user set-ups come in hand with important physical layer challenges: first, the cost of training overhead for channel estimation is proportional to the number of users and antennas, which translates into a larger complexity and also penalizes the achievable throughput  \cite{Marzetta2010}. Similarly, the complexity of the precoding design and the selection of the optimal set of users for transmission grows with the problem dimension (i.e., antennas and users) \cite{Muller2016,YoGo06}. Besides, precise channel characterization beyond simplified \ac{LoS} and one-ring approaches conventionally assumed in this context \cite{LuZe22,Han2020} are required to realistically predict network performance \cite{Chen20}. Finally, although the assumption of perfect \ac{CSI} availability is of widespread use in the \ac{XL}-\ac{MIMO} literature \cite{Marinello2020,GoLoCa21,Filho2022}, channel \textit{does} vary due to microscopic perturbations affecting \textit{both} the dominant specular components and the diffusely propagating multipath waves \cite{Demir22}, which introduces \ac{CSI} uncertainty and reduces the system performance \cite{TrHe13}.

Aiming to address the aforementioned challenges in the context of \ac{XL}-\ac{MIMO} systems, we analyze the problem of joint user scheduling and precoding in a realistic set-up. Specifically, we investigate the achievable \ac{SE} in a \ac{DL} multi-user scenario with a very large number of \ac{BS} antennas and users, on which user selection and precoding design are performed using outdated and noisy channel estimations. SW propagation features are integrated with a general channel model that combines the characteristics of the one-ring model with those of spatially correlated diffuse components \cite{Demir22}. We propose a channel training procedure in which only a subset of users participates in the training stage, thus reducing the training overhead and improving the achievable \ac{SE}. We design a joint user scheduling and precoding algorithm using imperfect and outdated channel estimates, and analyze the performance loss incurred due to such imperfect \ac{CSI}. Results show that the proposed algorithm outperforms other benchmark approaches based on \ac{SUS} when the cost of training overhead is considered.

\textit{Notation:} Lower and upper case bold letters denote vectors and matrices, while $\mathbb{C}^M$ is  complex vector space with dimension $M$; $\left ( \cdot \right )^\Tr$, $\left ( \cdot \right )^\He$, and $\trace\left ( \cdot \right )$ denote the transpose, Hermitian transpose, and trace operations; $\Real\{\cdot\}$ and $\Imag\{\cdot\}$ represent the real and the imaginary parts of a complex number; symbol $\sim$ reads as \emph{statistically distributed as}, and $\Exp\big[\cdot\big]$ is the statistical expectation; $\mathcal{N}_\mathbb{C}(\mathbf{0},\B{C})$ is a zero-mean Gaussian distribution with covariance $\B{C}$, and $\mathcal{U}[a,b)$ is a Uniform distribution within the interval; $\|\cdot\|$ is the Euclidean norm, and $|\cdot|$ a set cardinality.

\section{System model}
\label{sec:model}
Let us consider the \ac{DL} of an \ac{XL}-\ac{MIMO} setup, where the \ac{BS} equipped with an \ac{ELAA} with $M\gg 1$ elements sends data to $K$ single-antenna users. 
We assume a 2-D scenario where the \ac{BS} \ac{ELAA} is an \ac{ULA} centered at the origin of a circular coordinate system. In this scenario, each position is determined by a given radius $r_k$, which represents the distance to the center of the \ac{ULA}, and a given angle $\theta_k$. The signal arrives to the users through a \ac{LoS} channel component, and/or after several reflections. This is captured by the general channel model in \cite{Demir22}, so that the the $k$-th user channel vector $\B{h}_{\rm k}\in\mathbb{C}^M$ is expressed as
\begin{equation}
\B{h}_{\rm k}=\sum\nolimits_{s=1}^{S_k}\eul^{\imj \varphi_{k,s}}\hs_{\rm k,\rm s}+\hns_{\rm k},
\label{eq:channelModel}
\end{equation}
where $\varphi_{k,s}$ represents the phase shift due to microscopic perturbations; $\hns_{\rm k}$ is the contribution of diffuse scattering as the summation of numerous weak multipath waves; $\hs_{\rm k,\rm s}$ represent the dominant ${S_k}$ specular components\footnote{When using very large \ac{MIMO} and/or highly directive steerable antennas, specular paths need to be considered for realistic channel modeling \cite{metis2015}.} for user $k$, with $s=1$ denoting the \ac{LoS} component and $s=2\ldots S_k$ indicating the additional specular waves -- which are regarded as \ac{NLoS} paths. Assuming a block fading channel model of length $\tau_c$, with smaller duration that the channel coherence time, both $\varphi_{k,s}$ and $\hs_{\rm k,\rm s}$ are constant $\forall k,s$.  
On the contrary, $\varphi_{k,s}$ and $\hns_{\rm k}$ take independent realizations from one coherence block to another, but their statistics are available at the transmitter as follows: $\varphi_{k,s}\sim\mathcal{U}[0,2\pi)$  and  $\hns_{\rm k}\sim\mathcal{N}(\bf{0},\B{R}_{\rm k})$, with $\B{R}_{\rm k}$ capturing the spatial correlation of the antenna array and the long-term channel effects \cite{Demir22}. The power magnitudes of the specular and diffuse components are $\Omega_{\rm s}\approx\sum_{s=1}^{S_k}\|\hs_{\rm k,\rm s}\|^2$ and $\Omega_{\rm d}=\trace(\B{R}_{\rm k})$, respectively, with the power ratio being defined as $\kappa=\frac{\Omega_{\rm s}}{\Omega_{\rm d}}$, similar to the Rician $K$ parameter. 

Let $r_{k,s}$ and $\theta_{k,s}$ be the radius and the angle corresponding to either the user $k$ location for a \ac{LoS} component, or to the last reflection for a \ac{NLoS} specular path. 
Hence, the response vector $\hs_{\rm k,\rm s}$ in \eqref{eq:channelModel} reads as
\begin{equation}
\hs_{\rm k,\rm s}=\rho_{k,s}[\eul^{-\imj\frac{2\pi}{\lambda}r_{k,s,1}},\eul^{-\imj\frac{2\pi}{\lambda}r_{k,s,2}},\ldots,\eul^{-\imj\frac{2\pi}{\lambda}r_{k,s,M}}]^\Tr,
\end{equation}
where $\lambda$ is the signal wavelength and $\rho_{k,s}\in(0,1]$ is the attenuation of the specular component due to physical effects like path-loss, reflection or absorption. Finally, $r_{k,s,m}$ stands for the radius of the $m$-th element of the antenna array. Considering an spherical wave-front, this radius is
\begin{align}
r_{k,s,m}=r_{k,s}\sqrt{1-2m{d}_{k,s}\sin\theta_{k,s}+{d}_{k,s}^2m^2},  
\label{eq:antennaDistance}
\end{align}
with $m\in\left[-\tfrac{M}{2},\tfrac{M}{2}\right]$, $d_{k,s}=\frac{d}{r_{k,s}}$, and $d$ being the \ac{ULA} inter-antenna distance. 

The data symbols to be transmitted, denoted as $s_k\sim\mathcal{N}_\mathbb{C}(0,1)$, $k=1,\ldots,K$ are processed using linear precoding vectors ${\bf{p}}_{\rm k}\in\mathbb{C}^M$ that satisfy the power constraint $\sum_{k=1}^K\|{\bf{p}}_{\rm k}\|_2^2\leq P_\text{TX}$, being $P_\text{TX}$ the available transmit power at the \ac{BS}. At the user $k$, the received signal is
\begin{equation}
\label{eq4}
y_k=s_k{\bf{p}}_{\rm k}^\He{\bf{h}}_{\rm k}+\sum\nolimits_{j\neq k}s_j{\bf{p}}_{\rm j}^\He{\bf{h}}_{\rm k}+n_k,
\end{equation}
where the intended symbol $s_k$ is affected  by \ac{IUI} and \ac{AWGN} represented by $n_k\sim\mathcal{N}_\mathbb{C}(0,\sigma^2_n)$. By considering the \ac{IUI} (second term in \eqref{eq4}) as noise and assuming Gaussian signaling, the achievable \ac{SE} for user $k$ is defined as \cite{YoGo06}
\begin{align}
\label{eq:sumRate}
R_k=\left(1-\frac{\tau_p}{\tau_c}\right)\log_2\Bigg(1+\tfrac{|{\bf{p}}_{\rm k}^\He{\bf{h}}_{\rm k}|^2}{\sigma^2_n+\sum_{j\in\mathcal{S},j\neq k}|{\bf{p}}_{\rm j}^\He{\bf{h}}_{\rm k}|^2}\Bigg),
\end{align}
where $\tau_p$ is the number of channel uses employed for training on each coherence block.
The goal is to find the set of users $\mathcal{S}\subseteq\{1,\ldots,K\}$, and their corresponding precoders, that can be served in the same time-frequency resource, so that the achievable \ac{SE} is maximized, i.e.,
\begin{equation}
\argmax_{\{{\bf{p}}_{\rm k}\}_{k\in\mathcal{S}}} \sum\nolimits_{k\in\mathcal{S}}R_k\quad\text{s.t.}\quad\sum\nolimits_{k\in\mathcal{S}}\|{\bf{p}}_{\rm k}\|_2^2\leq P_\text{TX}.
\label{eq:problemForm}
\end{equation}

\section{User scheduling with imperfect CSI}

It is clear from \eqref{eq:problemForm} that the system performance strongly depends on the determination of the user set $\mathcal{S}$. While the consideration of instantaneous \ac{CSI} fully available at the BS is a typical assumption in the literature \cite{GoLoCa21,Filho2022,Souza2022,Ribeiro2021}, channel uncertainties due to channel aging and imperfect estimation also impact on this choice \cite{Chen20}. Besides, the cost of channel training reduces the achievable \ac{SE} as a pre-log factor in \eqref{eq:sumRate}, an issue usually neglected in the related literature. With all these aspects in mind, we tackle the problem of user scheduling under imperfect \ac{CSI}. This is accomplished in two steps: first define an \textit{equivalent channel gain} metric to estimate user priorities based on the channel estimates available at the \ac{BS}, and next proceed to the joint precoding and scheduling design based on imperfect CSI.

\subsection{Equivalent channel gains}
According to the block fading channel model described in the previous section, the diffuse component $\hns_{\rm k}$ and phase shifts $\varphi_{k,s}$ in \eqref{eq:channelModel} are unknown to the BS \cite{Demir22}. Hence, the expected channel gain for user $k$ is estimated as:
%
%
\begin{align}
	&\Exp\big[\|\B{h}_{\rm k}\|_2^2\big]=\sum\nolimits_{s=1}^{S_k}\|\hs_{\rm k,\rm s}\|_2^2+2\Exp\big[\Real\{\hns_{\rm k}^\He\hs_{\rm k,\rm s}\eul^{\imj \varphi_{k,s}}\}\big]\nonumber\\
	&+2\sum\nolimits_{z>s}^{S_k}\Exp\big[\Real\{\hs_{\rm k,\rm z}^\He\hs_{\rm k,\rm s}\eul^{\imj (\varphi_{k,s}-\varphi_{k,z})}\}\big]+\Exp\big[\|\hns_{\rm k}\|_2^2\big]\nonumber\\
	&\overset{(a)}{=}\sum\nolimits_{s=1}^{S_k}\|\hs_{\rm k,\rm s}\|_2^2+2\Exp\big[\Real\{\hns_{\rm k}^\He\hs_{\rm k,\rm s}\eul^{\imj \varphi_{k,s}}\}\big]+\trace(\B{R}_{\rm k})\nonumber\\
	&\overset{(b)}{=}\sum\nolimits_{s=1}^{S_k}\|\hs_{\rm k,\rm s}\|_2^2+\trace(\B{R}_{\rm k}).
	\label{eq:avgChannelGain}
\end{align}
In \eqref{eq:avgChannelGain}, equality (\textit{a}) comes from the fact that expectations for cross-products are $\sim0$ because of phase uncertainty: note first that small values for $\hs_{\rm k,\rm z}^\He\hs_{\rm k,\rm s}$, $s\neq z$ are likely; second, even if these products are non-negligible, we obtain
\begin{align}
	\Exp\big[\Real\{\hs_{\rm k,\rm z}^\He\hs_{\rm k,\rm s}&\eul^{\imj (\varphi_{k,s}-\varphi_{k,z})}\}\big]=\notag\\
	&\Real\{\hs_{\rm k,\rm z}^\He\hs_{\rm k,\rm s}\}\Exp\big[\Real\{\eul^{\imj (\varphi_{k,s}-\varphi_{k,z})}\}\big]\label{eq:gainTerm}\\
	&-\Imag\{\hs_{\rm k,\rm z}^\He\hs_{\rm k,\rm s}\}\Exp\big[\Imag\{\eul^{\imj (\varphi_{k,s}-\varphi_{k,z})}\}\big]\notag.
\end{align}
Focusing now on the first term in \eqref{eq:gainTerm} (a similar rationale applies to the second one), we obtain
\begin{align*}
	&\Exp\big[\Real\{\eul^{\imj (\varphi_{k,s}-\varphi_{k,z})}\}\big]=\Exp\big[\cos(\varphi_{k,s})\cos(\varphi_{k,z})\big]\\
	&+\Exp\big[\sin(\varphi_{k,s})\sin(\varphi_{k,z})\big]=\\
	&\left(\tfrac{1}{2\pi}\int_{0}^{2\pi}\cos(\varphi_{k,s})\right)^2+\left(\tfrac{1}{2\pi}\int_{0}^{2\pi}\sin(\varphi_{k,s})\right)^2=0,
\end{align*}
where last equality comes from the i.i.d. property of the phase shifts related to the microscopic perturbations for the channel paths in \eqref{eq:channelModel}. For equality (\textit{b}) in \eqref{eq:avgChannelGain} we use
\begin{align}
	\Exp\big[\Real\{\hns_{\rm k}^\He\hs_{\rm k,\rm s}\eul^{\imj \varphi_{k,s}}\}\big]&=
	\Exp\big[\Real\{\hns_{\rm k}^\He\hs_{\rm k,\rm s}\}\big]\Exp\big[\cos(\varphi_{k,s})\big]\notag\\
	&-\Exp\big[\Imag\{\hns_{\rm k}^\He\hs_{\rm k,\rm s}\}\big]\Exp\big[\sin(\varphi_{k,s})\big]=0,\notag
\end{align}
that is a consequence of $\hns_{\rm k}\sim\mathcal{N}(\bf{0},\B{R}_{\rm k})$. 
%

Now, when the number of users $K$ and antennas $M$ are large, determining $	\Exp\big[\|\B{h}_{\rm k}\|_2^2\big]$ $\forall k$ is costly. Therefore, we propose to approximate \eqref{eq:avgChannelGain} by the following expression
\begin{equation}
	\Exp\big[\|\B{h}_{\rm k}\|_2^2\big]\approx M\left(1+\frac{1}{\kappa}\right)\sum\nolimits_{s=1}^{S_k}\rho_{k,s}^2=g_k.
	\label{eq:gain}
\end{equation}
Note that the \textit{equivalent} channel gain $g_k$ provides an approximation to the expected channel gains based on the long-term \ac{CSI} available at the \ac{BS}. Building on this idea, we will develop an iterative method to jointly select the users to be served by the \ac{BS} and design their precoding vectors for \ac{DL} transmission.

\subsection{Imperfect CSI scheduling and precoding}

Alg. \ref{alg:Scheduler}, presents the proposed imperfect-CSI scheduling and precoding (ISP) algorithm, which is based on the following rationale: During data transmission over the $n$-th coherence block, users to be served in the block $n+1$ are scheduled one by one. The algorithm starts determining the expected channel gains $g_k$ $\forall k \in \{1,\ldots,K\}$ as in \eqref{eq:gain}, using the long-term channel statistics, and scheduling the one with the largest gain. In subsequent iterations, the $g_k$ values are refined with a penalty factor accounting for the \ac{IUI} caused by previously scheduled users and, again, the user with the largest gain is scheduled. This penalty factor also considers the \ac{CSI} available at the \ac{BS} for the $n$-th coherence block. To further detail Alg. \ref{alg:Scheduler}, let $\mathcal{S}^{(\ell)}$ be the set of scheduled users at iteration $\ell$, with, $|\mathcal{S}^{(\ell)}|=\ell$. Therefore, the {equivalent} channel gain for a user $k\notin \mathcal{S}^{(\ell)}$ at iteration $\ell$, accounting for \ac{IUI}, is given by 
\begin{equation}
	g_k^{(\ell)}=g_k-\sum\nolimits_{s=1}^{S_k}\hs_{\rm k,\rm s}^\He\B{F}^{(\ell)}\hs_{\rm k,\rm s}-\trace\left(\B{R}_{\rm k}\B{F}^{(\ell)}\right),
	\label{eq:gupadte}
\end{equation}
where $\B{F}^{(\ell)}=\sum_{i=1}^\ell\B{f}_{\rm k_i}^{(i)}(\B{f}_{\rm k_i}^{(i)})^{\He}$ with $\B{f}_{\rm k_i}^{(i)}$ unit-norm precoders for the user $k_i\in\mathcal{S}^{(\ell)}$, selected at iteration $i$. The update of $g_k^{(\ell)}$ in \eqref{eq:gupadte} relies on computations similar to those in \eqref{eq:avgChannelGain} but the spatial correlation matrix $\B{R}_{\rm k}$ and the matrix of precoders $\B{F}^{(\ell)}$ must be determined based on imperfect \ac{CSI}. An expression for the spatial correlation matrix $\B{R}_{\rm k}$ is given in the following lemma. For the sake of notation simplicity,  index $k$ is dropped.

\begin{lemma}\label{thm}
Let $[\B{R}]_{m,n}$ be the element at row $m$ and column $n$ of the spatial correlation matrix $\B{R}$ given by 
\begin{equation}
\label{eqR}
[\B{R}]_{m,n}=\beta\int \eul^{-j\frac{2\pi}{\lambda}r_m}\eul^{j\frac{2\pi}{\lambda}r_n}f(\theta)d\theta,
\end{equation}
where $\beta$ is the average gain, $r_m$ and $r_n$ are the distances corresponding to the $m$ and $n$ antennas, and $f(\theta)$ is the probability density function. Then, $[\B{R}]_{m,n}$ can be computed in closed-form as
\begin{align}
[\B{R}]_{m,n}&=\tfrac{\beta\eul^{\imj\left(a-\frac{b^2}{4c}\right)}\sqrt{\pi}(1+\imj)}{4\varphi\sqrt{2c}}\bigg[\phi\left(\tfrac{1-\imj}{\sqrt{2}}\big(\sqrt{c}\varphi+\tfrac{b}{2\sqrt{c}}\big)\right)\notag\\
&-\phi\left(\tfrac{1-\imj}{\sqrt{2}}\big(-\sqrt{c}\varphi+\tfrac{b}{2\sqrt{c}}\big)\right)\bigg].
\label{eq:covUniform}
\end{align}
where $\phi(\cdot)$ is the error function, and
\begin{align*}
a&=\tfrac{2\pi}{\lambda}\big[(m-n)d\sin(\vartheta)+\tfrac{(n^2-m^2)d^2}{2r}\cos^2(\vartheta),\label{eq:aux}\\
b&=\tfrac{2\pi}{\lambda}\cos(\vartheta)[(m-n)d-\tfrac{(n^2-m^2)d^2}{r}\sin(\vartheta)],\\
c&=\tfrac{2\pi}{\lambda}\tfrac{(n^2-m^2)d^2}{2r}\sin^2(\vartheta).
\end{align*}

\end{lemma}
\begin{proof}\label{proof-th1}
See Appendix. 
\end{proof}

To compute the matrix of precoders $\B{F}^{(\ell)}$ in \eqref{eq:gupadte}, we model the  $k$-th user current channel state in terms of a past state plus an innovation component \cite{TrHe13}, as follows
\begin{equation}
	\B{h}_{\rm k}[n+1]=\alpha\B{h}_{\rm k}[n]+\B{z}_{\rm k}[n],
	\label{eq:errorModel}
\end{equation}
where $\B{h}_{\rm k}[n]$ is the channel state at the $n$-th coherence block, $\alpha$ is a temporal correlation parameter, and $\B{z}_{\rm k}[n]\sim\mathcal{N}(\B{0},(1-\alpha^2)\B{R}_{\rm \B{z}})$ is a channel innovation term uncorrelated with $\B{h}_{\rm k}[n]$. 
By considering the estimation error incurred during the $n$-th block,  $\breve{\B{h}}_{\rm k}[n]$, such that $\B{h}_{\rm k}[n]=\hat{\B{h}}_{\rm k}[n]+\breve{\B{h}}_{\rm k}[n]$, \eqref{eq:errorModel} results in
\begin{equation}
\B{h}_{\rm k}[n+1]=\alpha\hat{\B{h}}_{\rm k}[n]+\alpha\breve{\B{h}}_{\rm k}[n]+\B{z}_{\rm k}[n]=\alpha\hat{\B{h}}_{\rm k}[n]+\B{e}_{\rm k}[n],
\label{eq:himperfect}
\end{equation}
where $\B{e}_{\rm k}[n]$ is a zero-mean error term. Assuming least-squares channel estimation with an orthogonal training sequence, the covariance matrix of $\B{e}_{\rm k}[n]$ is $\B{R}_{\rm \B{e}}=\alpha^2\frac{\sigma_n^2}{ P_\text{TX}}\B{I}+(1-\alpha^2)\B{R}_{\rm \B{z}}$. Finally, for the spatial correlation of the innovation term, we derive $\Exp\big[\B{h}_{\rm k}\B{h}_{\rm k}^\He\big]$  in a way similar to \eqref{eq:avgChannelGain}, i.e.,
\begin{equation}
	\B{R}_{\rm \B{z}}=\sum\nolimits_{s=1}^{S_k}\hs_{\rm k,\rm s}\hs_{\rm k,\rm s}^\He+\B{R}_{\rm k}.
\end{equation}

\begin{algorithm}[t]
	\caption{I-CSI Scheduling-Precoding  (ISP)}\label{alg:Scheduler}
	\begin{algorithmic}[1]
		\small
		\STATEx \textbf{During data transmission stage for coherence block $n$}
		\STATE $\mathcal{S}^{(0)}\gets\emptyset$,  $\ell \gets 0$
		\STATE $g_k$, $\forall k \gets$ initialization with \eqref{eq:gain}
		\REPEAT
		\REPEAT
		\STATE $k\gets \max_{i\notin\mathcal{S}^{(\ell)}}g_{i}^{(\ell)}$
		\STATE $g_{k}^{(\ell)}\gets$ update using \eqref{eq:gupadte} 
		\STATE $q\gets \max_{i\notin\mathcal{S}^{(\ell)}}g_{i}^{(\ell)}$
		\UNTIL{$q=k$}
		\STATE  $\ell \gets \ell+1$
		\STATE $\mathcal{S}^{(\ell)}\gets\mathcal{S}^{(\ell-1)}\cup\{k\}$
		\STATE $\B{f}^{(\ell)}_{\rm k}\gets$ Compute ZF precoders \eqref{eq:ZF-ICSI}, $\forall k\in\mathcal{S}^{(\ell)}$
		\STATE ${\bf{p}}^{(\ell)}_{\rm k}\gets\B{f}^{(\ell)}_{\rm k}p_{\rm k}^{(\ell)} $, with waterfilling power allocation 
		\STATE $\tau_p^{(\ell)}=(|\mathcal{S}^{(\ell)}|+|\mathcal{G}|)\dot\tau$
		\STATE $\sum_{k\in\mathcal{S}}R_k$ with $\B{p}^{(\ell)}_{\rm k},\hat{\B{h}}_{\rm k}[n],\forall k$, and $\tau_p^{(\ell)}$ [cf. \eqref{eq:sumRate}]
		\UNTIL{$|\mathcal{S}^{(\ell)}|=K$ or performance metric decreases}
		\STATE Determine $\mathcal{G}$, $\mathcal{S}\gets\mathcal{S}^{(\ell)}$
		\STATEx \textbf{Training stage for coherence block $n+1$}
		\STATE Estimate channels for users $k\in\mathcal{S}\cup\mathcal{G}$
		\STATE Compute ZF precoders $\forall k\in\mathcal{S}$, with $\hat{\B{h}}_{\rm k}[n+1]$
		\STATEx \textbf{Data transmission stage for coherence block $n+1$}
	\end{algorithmic}
\end{algorithm}

Based on this \ac{CSI} uncertainty model, the \ac{ZF} precoding vectors in \eqref{eq:gupadte} $\forall k\in\mathcal{S}^{(\ell)}$ are computed as
\begin{equation}
[\tilde{\B{f}}_{\rm k_1}^{(\ell)},\ldots,\tilde{\B{f}}_{\rm k_\ell}^{(\ell)}]^T=(\hat{\B{H}}^{(\ell)}(\hat{\B{H}}^{(\ell)})^\He)^{-1}\hat{\B{H}}^{(\ell)},
\label{eq:ZF-ICSI}
\end{equation}
where $\hat{\B{H}}^{(\ell)}=[\hat{\B{h}}_{\rm k_1}[n],\ldots,\hat{\B{h}}_{\rm k_\ell}[n]]^T$ contains the outdated channel estimates. Then, the unit-norm precoders $\B{f}_{\rm k}^{(\ell)}$ are readily obtained as $\B{f}_{\rm k}^{(\ell)}=\tilde{\B{f}}_{\rm k}^{(\ell)}/\|\tilde{\B{f}}_{\rm k}^{(\ell)}\|$. These precoders, together with the updated power allocation $p_k$, $\forall k$ and pre-log factors, are used to compute the achievable \ac{SE} in \eqref{eq:sumRate} at each iteration. ISP iterations end when selecting a new user does not improve the achievable \ac{SE}, returning the final scheduling set $\mathcal{S}$. Interestingly, as only the long-term channel conditions and the $n$-th coherence block \ac{CSI} are required for the proposed procedure, we determine the set $\mathcal{S}$ \textit{prior} to the training procedure of the $n+1$ coherence block. We also propose to update the \ac{CSI} at the $n+1$ coherence block training stage only for those users with good channel conditions, i.e. either those scheduled for transmission ($k\in\mathcal{S}$), or those potential candidates for being scheduled ($k\in\mathcal{G}$), such that $\mathcal{S}\cap \mathcal{G}=\emptyset$.  Candidate users in $\mathcal{G}$ may be selected according to different criteria; for instance, users satisfying $g_k\geq \nu$, where $\nu$ is a threshold for the expected channel gains. This reduces the pre-log factor in \eqref{eq:sumRate}; specifically, if $\tau_p$ scales linearly with the number of users \cite{Demir22}, then
\begin{equation}
	\frac{\tau_c-\tau_p}{\tau_c}=\frac{\tau_c-(|\mathcal{S}|+|\mathcal{G}|)\dot\tau}{\tau_c},
\end{equation}
where $\dot\tau$ is the number of channel uses employed to acquire the \ac{CSI} of an individual user. 
Once the channel for the users belonging to $\mathcal{S}$ are estimated, the \ac{BS} computes the \ac{ZF} precoders based on the available imperfect \ac{CSI}, i.e., $\hat{\B{h}}_{\rm k}[n+1]$. Overall, the key benefits of ISP are the reduction of the training overhead, and the removal of the scheduling procedure from the data transmission stage. 

\section{Numerical Results}

The performance of ISP has been assessed through computer experiments. Unless explicitly stated, the simulation parameters are those in Table \ref{tab:Sim1} which correspond to an urban environment with users moving at a relatively high speed. For this setup, the \textit{critical distance} is $135$m \cite{GoLoCa21}, and we consider a temporal correlation factor $\alpha=J_0(2\pi f_d T_s\tau_s)$, with $J_0(\cdot)$ the Bessel function of the first kind and zero-th order, $f_d$ the Doppler frequency, $T_s$ the sampling period, and $\tau_s$ the \ac{CSI} delay in terms of samples \cite{Cl68}.

\begin{table}[t]
	\centering
	\caption{{Simulation parameter settings.}}\label{tab:Sim1}
	\setlength{\tabcolsep}{5pt}
	\def\arraystretch{1.2}
        \vspace*{-0.2cm}
	\begin{tabular}{|l|l|l|l|}
		\hline		
		{\textbf{Parameter}} & {\textbf{Value}} & {\textbf{Parameter}} & {\textbf{Value}} 		\\ \hline\hline
		Channel realiz.		&	100 &
		$\#$ of users & 		$K=200$ \\ \hline
		Specular comp. & 		$S_k=4$ 	&	
		Wavelength & 			$\lambda=0.15$ m 				\\ \hline
		Sampling freq. & 			$f_s=1$ MHz				
		& Antenna dist. & $d=\frac{\lambda}{2}$ m 	\\ \hline
		$\#$ of antennas & 	$M=200$  					&
		Power ratio & 	$\kappa=2$  						\\ \hline
		Angular range &		 	$[-\frac{\pi}{4}, \frac{\pi}{4}]$ rad &
		Distance range &		 	$[40,230]$ m \\ \hline 
		Angular std dev. &		 	$\sigma_\delta=10^\circ$ &
		Block len. (samp.) & $\tau_c=10000$ \\ \hline
		CSI delay (samp.) & $\tau_s=10000$  &
		User speed & $v=30$ km/h\\ \hline
	\end{tabular}
\end{table}
We first evaluate in Fig. \ref{fig:diffuse} the impact of different channel configuration parameters in \eqref{eq:channelModel} under the assumption of perfect \ac{CSI}, and use the optimal \ac{SUS} procedure in \cite{YoGo06} for benchmarking purposes. We consider different number of specular paths $S_k$ and power ratios $\kappa$. Experiments show that increasing the number of paths $S_k$ leads to larger channel gains and, accordingly, improves the achievable \ac{SE}. A similar effect arises when increasing the value of $\kappa$, i.e., the power received from the specular paths. We see that the worst performance is obtained for the \ac{LoS} case. In other words, richer channels facilitate obtaining better throughputs, as the improvements in channel gains dominate over the potential increase in interference to other users. 

\begin{figure}[t]
	\centering
	\includegraphics[width=.9\columnwidth]{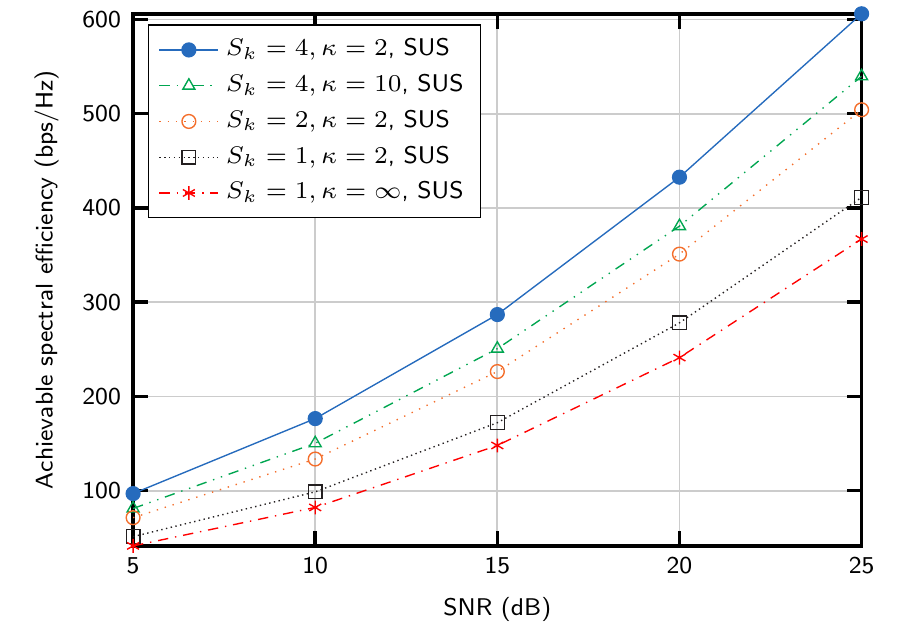}
	\caption{Achievable \ac{SE} vs. SNR (dB) for a different number of paths $S_k$ and power ratios $\kappa$. The remaining parameter values are given in Table \ref{tab:Sim1}. } 
	\label{fig:diffuse}
  \vspace{-2mm}
\end{figure}

In Fig. \ref{fig:imperfect} we assess the performance losses due to imperfect \ac{CSI} according to the channel model in \eqref{eq:himperfect} and $\tau_s=\{2000,10000\}$ samples. For the estimation error, we assume that the SNR remains constant during both training and data transmission. We compare with the perfect \ac{CSI} scenario, and also with the hypothetical situation where the scheduling procedure is performed with the imperfect \ac{CSI}, but the \ac{ZF}-precoders at the data transmission stage use the \textit{true} channels. This way, the performance losses caused by imperfect \ac{CSI} are separated for \textit{i}) scheduling and \textit{ii}) \ac{ZF} precoding. This latter curve is labeled as {ISP-P} for reference purposes. Remarkably, when disregarding the losses due to imperfect \ac{ZF} precoding. the ISP procedure operating with outdated \ac{CSI} achieves a performance similar to the perfect \ac{CSI} case.


\begin{figure}[t]
	\centering
	\includegraphics[width=.9\columnwidth]{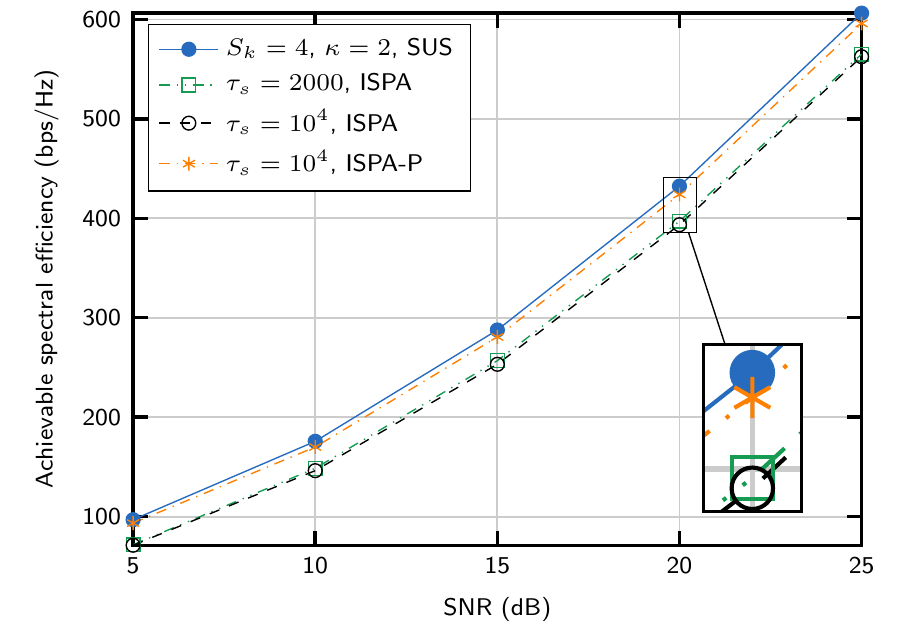}
	\caption{Achievable \ac{SE} vs. SNR (dB) for $S_k=4$, $\kappa=2$ and different \ac{CSI} delays $\tau_s$. Parameter values are given in Table \ref{tab:Sim1}. }
	\label{fig:imperfect} 
  \vspace{-2mm}
\end{figure}

\begin{figure}[h]
	\centering
	\includegraphics[width=.9\columnwidth]{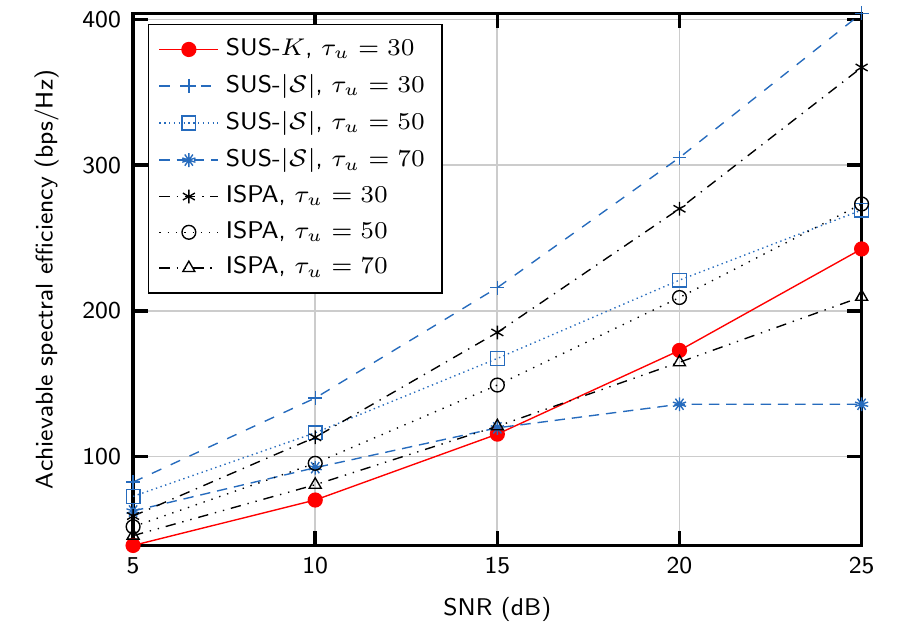}
	\caption{Achievable \ac{SE} vs. SNR (dB) for  $S_k=4$, $\kappa=2$, $\tau_s=10^4$, and several training overheads $\dot\tau$. Parameter values are given in Table \ref{tab:Sim1}. } 
	\label{fig:overhead}
  \vspace{-2mm}
\end{figure}

Fig. \ref{fig:overhead} studies the effect of the training overhead by modifying the amount of time $\dot\tau$ necessary to estimate an individual user. For completeness, we include the cases where \textit{all} users have their channels estimated, \textit{SUS-$K$}, and a genie-aided scheme where the training stage just updates the channels for the set of scheduled users, \textit{SUS-$|\mathcal{S}|$}. The training lengths considered are $\dot\tau\in\{30,50,70\}$ samples. Note that $\dot\tau=30$ is a very optimistic value for $M=200$ \cite{Demir22}. The curves for ISP assume that, in addition to the scheduled users in $\mathcal{S}$, the non-scheduled candidate users with large gains in $\mathcal{G}$ have their channels estimated. We observed in our experiments that the number of users  in $\mathcal{G}$ (i.e., those considered for scheduling but finally not served) is, on average, less than $10$ for $\dot\tau=30$, and smaller values are obtained for larger $\dot\tau$. Therefore, we consider $|\mathcal{G}|=15$ for the ISP configuration.

We observe in Fig. \ref{fig:overhead} a significant reduction in the achievable \ac{SE} caused by the training overhead and imperfect \ac{CSI}. Observe also that, if the channels for all $K$ users are estimated, the performance dramatically drops even in the optimistic scenario with $\dot\tau=30$. In turn, when the genie-aided approach is considered, the best result is achieved for $\dot\tau=30$. Nevertheless, as typical scheduling schemes ignore the performance impact of training overhead, when $\dot\tau=50$ the performance of the genie-aided approach and ISP are similar. We see that when $\dot\tau=70$ ISP provides better results for SNRs above $15$dB. In general, ISP selects users exhibiting large expected channel gains in \eqref{eq:gain}, which are reduced when \ac{CSI} accuracy is low, thus promoting a conservative users selection compared to conventional schemes. Finally, as $\dot\tau$ and the training overhead increases with both $M$ and $K$, we expect that the benefits of employing ISP increase for larger values of $M$ and/or $K$.



\section{Conclusion}
The design of user selection schemes for XL-MIMO systems under practical constraints is fundamental for their successful deployment. Assuming a realitic channel model beyond the conventional \ac{LoS} assumption, together with \ac{CSI} uncertainties and training overheads, we addressed this problem for the first time in the literature under SW propagation. The proposed algorithm incorporates this knowledge into the design, showing that the performance of the scheduler is similar to that achieved with solutions designed for perfect \ac{CSI}.

\begin{appendix}
\label{sec:correlation}

%
Starting from \eqref{eqR}, we consider a local scattering model where we have a nominal angle $\vartheta$ plus a random component $\delta$, i.e., $\theta=\vartheta+\delta$ \cite{Demir22}. We consider small values for $\delta$, i.e., $\delta<\frac{\pi}{12}$, since values of $\delta$ in this range are practical for urban environments, and smaller angles even apply to rural areas. To determine \eqref{eqR} we resort to the Fresnel approximation \cite{Sh62,CuDa22} 
\begin{equation}
r_m\approx r-md\sin(\vartheta+\delta)+\frac{m^2d^2\cos^2(\vartheta+\delta)}{2r}.
\end{equation}
Using trigonometrical identities, the integral in \eqref{eqR} is
\begin{align}
[\B{R}]_{m,n}&=\beta\int \eul^{{\imj\frac{2\pi}{\lambda}}\big[(m-n)d[\sin(\vartheta)\cos(\delta)+\cos(\vartheta)\sin(\delta)]}\notag\\
&\times\eul^{{\imj\frac{2\pi}{\lambda}}\frac{(n^2-m^2)d^2}{2r}[\cos(\vartheta)\cos(\delta)-\sin(\vartheta)\sin(\delta)]^2}f(\delta)d\delta\notag\\
&=\beta\int\eul^{\imj(\tilde{a}+\tilde{b}+\tilde{c})}f(\delta)d\delta, \label{eq:R}
\end{align}
where we have rearranged the terms as
\begin{align*}
\tilde{a}&=\tfrac{2\pi}{\lambda}\big[(m-n)d\sin(\vartheta)\cos(\delta)+\tfrac{(n^2-m^2)d^2}{2r}\cos^2(\vartheta)\cos^2(\delta)]\notag\\
\tilde{b}&=\tfrac{2\pi}{\lambda}\cos(\vartheta)\sin(\delta)[(m-n)d-\tfrac{(n^2-m^2)d^2}t{r}\cos(\delta)\sin(\vartheta)]\notag\\
\tilde{c}&=\tfrac{2\pi}{\lambda}\tfrac{(n^2-m^2)d^2}{2r}\sin^2(\vartheta)\sin^2(\delta).
\end{align*}
When $d/r\approx 0$, results reduce to the far-field scenario.

\end{appendix}

\bibliographystyle{IEEEtran}
\bibliography{references}

\end{document}